\documentclass[a4paper,12pt]{article}

\usepackage{hyperref}
\usepackage{amsmath}
\usepackage{amssymb}
\usepackage{amsthm}
\usepackage{a4wide}
\usepackage{color}

\usepackage{graphicx}

\usepackage[font=footnotesize, labelfont={sf,bf}, margin=1cm]{caption}

\theoremstyle{definition}
\theoremstyle{plain}\newtheorem{obs}{Observation}[section]
\theoremstyle{plain}\newtheorem{thm}{Theorem}[section]
\theoremstyle{plain}
\theoremstyle{plain}\newtheorem{cor}{Corollary}[section]
\theoremstyle{plain}\newtheorem{remark}{Remark}[section]

\newcommand\fitch{\mathbf{F}}
\newcommand\rt{\rho} 
\newcommand\stateset{\mathbf{C}}
\newcommand\ch{\Delta}
\newcommand\dMP{\ensuremath{d_{\textsc{MP}}}}
\newcommand{\dTBR}{\ensuremath{d_{\textsc{TBR}}}}
\newcommand{\tbr}{\textsc{tbr}}
\newcommand{\maf}{\textsc{maf}}
\newcommand{\spr}{\textsc{spr}}

\newcommand{\steven}{\textcolor{black}}

\newcommand{\kelk}{\textcolor{black}}

\author{Steven Kelk\thanks{Department of Knowledge Engineering, Maastricht University, P.O. Box 616, 6200 MD Maastricht, Netherlands. \kelk{(Kelk and Wu are corresponding authors).}},  Mareike Fischer\thanks{Institut f\"ur Mathematik und Informatik,
Walther-Rathenau-Stra{\ss}e 47,
17487 Greifswald, Germany.}, Vincent Moulton, Taoyang Wu\thanks{Moulton and Wu are both affiliated to: School of Computing Sciences, University of East Anglia, Norwich, NR4 7TJ,
United Kingdom.}}

\title{Reduction rules for the maximum parsimony 
distance on phylogenetic trees}
\date{\today}

\begin{document}
\maketitle
\begin{abstract}
\noindent  
In phylogenetics, distances are often used to measure the incongruence between a pair of phylogenetic trees that are reconstructed by different methods or using different regions of genome. Motivated by the maximum parsimony principle in tree inference, we recently introduced the maximum parsimony (MP) distance,  which enjoys various attractive properties due to its  connection with several other well-known tree distances, such as $\tbr$ and $\spr$. Here we show that computing the MP distance between two trees, a NP-hard problem in general, is fixed parameter
tractable in terms of the $\tbr$ distance between the tree pair. Our approach is based on two reduction rules--the chain reduction and the subtree reduction--that are widely used in computing $\tbr$ and $\spr$ distances. More precisely,  we show that reducing chains to length 4 (but not shorter) preserves the MP distance.  In addition, we describe a 
generalization of the subtree reduction which allows  the pendant subtrees to be rooted in different 
places, and show that this still preserves the MP distance. \steven{On a slightly different note we also show that Monadic Second Order Logic (MSOL), posited over an auxiliary graph structure known as the display graph (obtained by merging the two trees at their leaves), can be used to obtain an alternative proof that computation of MP distance is fixed parameter tractable in terms of $\tbr$-distance.} We conclude with an extended discussion in which we focus on similarities and differences between MP distance and TBR
distance and present a number of open problems. \steven{One particularly intriguing question, emerging from the MSOL formulation, is whether two trees with bounded MP distance induce display graphs of bounded treewidth.}
\end{abstract}




\section{Introduction}

Finding an optimal tree explaining the relationships 
of a group of species based on datasets at the genomic level is one of the important challenges in modern phylogenetics. First, there are various 
methods to estimate the ``best'' tree subject to 
certain criteria, such as e.g. Maximum Parsimony or 
Maximum Likelihood. However, different methods often 
lead to different trees for the same dataset, or the 
same method leads to different trees when 
different parameter values are used. Second, the  trees reconstructed from different regions of the genome might also be different, even when using the same criteria. In any case, 
when two (or more) trees for one particular set of 
species are given, the problem is to quantify how different 
the trees really are -- are they entirely different or 
do they agree concerning the placement of most species? 

In order to answer this problem, various distances have 
been proposed (see e.g.~\cite{steel1993distributions}). 
A relatively new one is the so-called 
Maximum Parsimony distance, or MP distance for short, 
which we denote $\dMP$ \cite{dMP-fischer,dMP-kelk,dMP-moulton}. 
This distance (which is a metric) is appealing in part due to the fact that it is closely 
related to the parsimony criterion for constructing 
phylogenetic trees, as well as to the Subtree Prune and Regraft ({\spr}) and Tree Bisection and Reconnection ({\tbr}) distances.  \kelk{Indeed, it is shown in \cite{dMP-moulton} that the unit neighbourhood of the MP distance is larger than those of the {\spr} and {\tbr} distances, implying that a hill-climbing heuristic search based on the MP distance will be less likely to be trapped in a local optimum than those based on the {\spr} or {\tbr} distances.} Recently, it has been shown that 
computing the MP distance 
is NP-hard \cite{dMP-fischer,dMP-kelk}
even for binary phylogenetic trees. 
For practical purposes it is therefore desirable to determine whether computation
of $\dMP$ is fixed parameter tractable (FPT). Informally, this asks whether $\dMP$ can be computed efficiently when $\dMP$ (or some other parameter of the input) is small, irrespective of the
number of species in the input trees. We refer to standard texts such as \cite{downey2013fundamentals} for more background on FPT. Such algorithms are used extensively in phylogenetics, see e.g. \cite{Whidden2014} for a recent
example.


An obvious approach to address this
question is to try to \emph{kernelize} the problem.
Roughly speaking, when given two trees, we seek to simplify 
them as much as possible without changing  
$\dMP$ so that we can calculate the 
distance for the simpler trees rather 
than the original ones. Standard procedures that have been used to 
kernelize other phylogenetic tree distances are the 
so-called subtree and chain reductions (see, for example, 
\cite{allen01,bordewich07, hickey08}). 
In this paper we show that the chain reduction preserves $\dMP$ and that 
chains can be reduced to length 4 (but not less). 
Moreover, we show that a certain generalized subtree reduction, namely one 
where the subtrees are allowed to have different 
root positions, also has this property, which 
extends a result in \cite{dMP-moulton}. Both reductions can be applied in polynomial
time.

These new results allow us to leverage the existing literature on {\tbr} distance.
Specifically, in
\cite{allen01} Allen and Steel showed
that {\tbr} distance, denoted \dTBR, is NP-hard to 
compute, by exploiting the essential equivalence of the problem with
the Maximum Agreement Forest (\maf) problem: \steven{they differ by exactly 1}. In the same article
they showed \kelk{(again utilizing the equivalence with {\maf})} that computation of \steven{$\dTBR$}
is FPT in parameter \steven{$\dTBR$}.
More specifically, it was shown that combining the subtree
reduction with the chain reduction (where chains are 
reduced to length 3, rather
than length 4 as we do here) is sufficient to 
obtain a reduced pair of trees where the number of species
is at most a \emph{linear} 
function of \steven{$\dTBR$}.
Careful
reading of the analysis in \cite{allen01} shows 
that a linear (albeit slightly larger) kernel is 
still obtained for
\steven{$\dTBR$}
if chains are 
reduced to length 4 rather than 3. More recently, 
in \cite{kelkFibonacci} an exponential-time 
algorithm was described and implemented which 
computes $\dMP$ in time $\Theta(\phi^{n} \cdot \text{poly}(n))$ 
where $n$ is the number of species in the trees and $\phi \approx 1.618...$ is the golden ratio. 
Combining the results of \cite{allen01,kelkFibonacci} with the main
results of the current paper (i.e. Theorems \ref{thm:reduc} and \ref{thm:gensubtree}) immediately yields the following theorem:\\
\\
\begin{thm}
\label{thm:fpttbr}
Let $T_1$ and $T_2$ be two unrooted binary 
trees on the same set of species $X$. Computation
of $\dMP(T_1,T_2)$ is fixed parameter tractable 
in parameter $\dTBR = \dTBR(T_1, T_2)$.  More 
specifically, $\dMP(T_1,T_2)$ can be computed 
in time $O( \phi^{c \cdot \dTBR} \cdot \emph{poly}(|X|) )$ where $\phi \approx 1.618...$
is the golden ratio and $c \leq 112/3$.
\end{thm}

The constant 112/3 is obtained by multiplying the bound on the size of the kernel given in \cite{allen01} ($28 \cdot \dTBR$)  by a factor 4/3, which adjusts for the fact that here chains are reduced to length 4 rather than 3.
Note also that Theorem~\ref{thm:fpttbr} does not require us to 
apply the generalized subtree reduction: the traditional subtree 
reduction together with the chain reduction is sufficient.

We now summarise the rest of the paper. In the next section we collect some necessary definitions and notations, including a brief description of Fitch's algorithm which our proofs extensively use. Then in the following three sections we establish the two reductions for the MP distance, that is, the chain reduction and the subtree reduction, and remark that
a theoretical variant of Theorem \ref{thm:fpttbr} could also be attained by leveraging
Courcelle's Theorem~\cite{courcelle1990,Arnborg91}, extending in a non-trivial way a technique introduced in \cite{kelk2015}.
\steven{Specifically, computation of $\dMP(T_1,T_2)$ can be formulated as a sentence of
Monadic Second Order Logic (MSOL) posited over an auxiliary graph structure known as the display graph. The display graph is obtained by (informally) merging the two trees at their leaves. Crucially, the length
of the sentence, and the treewidth of the display graph, are shown to be both bounded as a function of $\dTBR$.}

 We end with an extended
discussion in which we focus on similarities and differences between MP distance and TBR
distance. From a theoretical perspective the two distances sometimes behave rather differently but in practice {$\dMP$} and {$\dTBR$} are often very close indeed.
The major open problem that remains is whether
computation of $\dMP$ is FPT when parameterized by itself. \steven{One possible route to this result is via a
strengthened MSOL formulation, but this requires a number of challenging questions to be answered. In particular,
can the treewidth of the display graph be bounded as a function of {\dMP} (rather than {\dTBR})? This in turn
is likely to require new structural results on the interaction between (large grid) minors in the display graph and
phylogenetic incongruency parameters.}

\section{Preliminaries}
\label{sec:pre}

\subsection{Basic definitions}

An \emph{unrooted binary phylogenetic tree} on a set of species (or, more abstractly, \emph{taxa}) $X$ is a connected, undirected
tree in which all internal nodes have degree 3 and the leaves are bijectively labelled
by $X$. For brevity we henceforth refer to these simply as \emph{trees}, and we often
use the elements of $X$ to denote the leaves they label. In some cases, we have to consider \emph{rooted binary phylogenetic trees} instead of unrooted ones. These trees have an additional internal node of degree 2. When referring to such trees, we will talk about \emph{rooted trees} for short.

 For two trees $T_1$ and $T_2$ on the same set of taxa $X$, we write $T_1 = T_2$ if there
is an isomorphism between the two trees that preserves the labels $X$. The expression $T|_{X'}$,
where $X' \subseteq X$, has the usual definition, namely: the tree obtained by taking the unique minimal spanning tree on $X'$ and then repeatedly suppressing any nodes of degree 2. 

A \emph{character} on $X$ is a surjective function $f: X \rightarrow \stateset$ where $\stateset$ is
a set of \emph{states}. Given a phylogenetic tree $T = (V,E)$ on $X$, and a character
$f$ on $X$, an \emph{extension} of $f$ to $T$ is a mapping $\overline{f}: V \rightarrow \stateset$
which extends $f$ i.e. for every $x \in X$, $f(x) = \overline{f}(x)$. The number of
\emph{mutations} induced by $\overline{f}$, denoted by $\ch(\overline{f})$, is defined to be the number of edges
$\{u,v\} \in E$ such that $\overline{f}(u) \neq \overline{f}(v)$. The \emph{parsimony score} of $f$ on $T$ (sometimes called the \emph{length}) is defined to be the
minimum, ranging over all extensions $\overline{f}$ of $f$ to $T$, of the number of mutations
induced by $\overline{f}$. This is denoted $l_{f}(T)$.  Following~\cite{wu2009refining}, an extension $\overline{f}$ that
achieves this minimum is called a \emph{minimum} extension (also known as an optimal extension, but here we reserve the word optimal for other use). This value can be computed in
polynomial time using dynamic programming. Fitch's algorithm is the most well-known
example of this. (We will
use Fitch's algorithm extensively in this article and give a brief description of its execution
in the next section). 

Given two trees $T_1$ and $T_2$ on $X$, the \emph{maximum parsimony distance} of
$T_1$ and $T_2$, denoted $\dMP = \dMP(T_1, T_2)$, is defined as
\[
\dMP(T_1, T_2) = \max_{f} | l_f(T_1) - l_f(T_2) |
\]
where $f$ ranges over all characters on $X$. A character $f$ that achieves this maximum
is called an \emph{optimal} character. In \cite{dMP-fischer,dMP-moulton} it is proven that
$\dMP$ is a metric.

Note that in this manuscript, we also compare $\dMP$ to the well-known \emph{Tree Bisection and Reconnection (TBR)} distance, denoted $d_{TBR}$. Recall that a TBR move is performed as follows: Given an unrooted binary phylogenetic tree, delete one edge and suppress all resulting nodes of degree 2. Of the two trees now present, if they consist of at least two nodes, pick an edge and place a degree-2 node on it and choose it; else if either one only consists of one leaf, choose this leaf. Now connect the two chosen nodes with a new edge. This completes the TBR move. Note that $d_{\tbr}(T_1,T_2)$ is defined as the minimum number of TBR moves needed to transform $T_1$ into $T_2$.  In \cite{dMP-fischer,dMP-moulton} it is proven that
$\dMP(T_1,T_2) \leq \dTBR(T_1,T_2)$ for all trees $T_1, T_2$, with both articles listing examples
where the inequality is strict.

A concept which often occurs when discussing tree distances is the so-called \emph{agreement forest} abstraction. Recall that, given two trees $T_1$ and $T_2$ on $X$, an \emph{agreement forest} is a partition of $X$ into non-empty subsets $X_1,\ldots,X_k$, such that $T_1|_{X_i}$ and $T_2|_{X_i}$ are isomorphic for all $i$, and such that the subtrees $T_t|_{X_i}$ and $T_t|_{X_j}$ are node disjoint subtrees of $T_t$ for all $i$ and $j \in \{1,\ldots, k\}$ and for $t=1,2$. An agreement forest with a minimum number of components is called a \emph{Maximum Agreement Forest}, or MAF for short. In \cite{allen01} it was proven that $\dTBR$ is equal
to the number of components in a MAF, minus one.

The last concept we need to recall is \emph{fixed parameter tractability} (FPT).  An algorithm is fixed parameter tractable in parameter $k$
if its running time has the form $g(k) \cdot \text{poly}(n)$ where $n$ is the size of the input (here we take $n=|X|$) and $g$ is some (usually exponential) computable function that depends
only on $k$. For distances on trees it is quite usual to take the distance itself as the parameter, but other parameters can be chosen, and this is the approach we take in this article (i.e. we
parameterize computation of $\dMP$ in terms of $\dTBR$). For more formal background on FPT we refer the reader to \cite{downey2013fundamentals}.

\steven{We defer a number of definitions (concerning treewidth and display graphs) until later in the article.}

\subsection{Fitch's algorithm}
For a given character $f$ on $T$, Fitch's algorithm \cite{Fitch} is a well-known polynomial-time algorithm for computing $l_f(T)$ and inferring a minimum extension of $f$ (see, e.g.~\cite{yang2011analysis}, for a recent application). It has a bottom-up phase followed by a top-down phase (actually, in the original paper, Fitch introduced a second top-down phase, but this is not needed in the present manuscript and is thus ignored here). It works on rooted trees, but the location of the root is not important for computation of $l_f(T)$, so we may root the tree by subdividing an arbitrary edge with a new node $\rt$ and directing all edges away
from this new node. (In particular, this ensures that the child-parent relation is well-defined). For each internal node $u$ of a rooted tree, let $u_l$ and $u_r$ refer to its two children. 

In the first phase, the algorithm constructs the {\em Fitch map} $\fitch: V(T)\to 2^\stateset \setminus \{\emptyset\}$ (induced by character $f$) that  assigns a subset of states to each of node $u$ of $T$  in the following  bottom-up approach:
\begin{enumerate}
\item For each leaf $x$, let $\fitch(x) = \{f(x)\}$. 
\item For each internal node $u$ (for which $\fitch(u_l)$ and $\fitch(u_r)$ have already been computed), let
\begin{equation}
\label{eq:fitch:map}
\fitch(u)=
\fitch(u_l)*\fitch(u_r)
=
\begin{cases}
\fitch(u_l) \cup \fitch(u_r)  & \mbox{if  $\fitch(u_l) \cap \fitch(u_r) = \emptyset$,} \\
\fitch(u_l)\cap \fitch(u_r) & \mbox{otherwise.}
\end{cases}
\end{equation}
\end{enumerate}
An internal node $u$ is called  a $\emph{union}$ node if the first case in Equation~(\ref{eq:fitch:map}) occurs (i.e., $\fitch(u_l) \cap \fitch(u_r) = \emptyset$), and  an \emph{intersection} node otherwise. The
value $l_f(T)$ is equal to the total number of union nodes in $T$. 

For later use, an extension $\overline{f}$ of $f$ on $T$ is called a {\em Fitch-extension}
if (i) $\overline{f}(u)\in \fitch(u)$ holds for all $u\in V(T)$, and (ii) for each non-leaf node $u$ of $V(T)$, we have either $\overline{f}(u)=\overline{f}(u_l)$ or $\overline{f}(u)=\overline{f}(u_r)$ (but not both) if $u$ is a union node, and $\overline{f}(u)=\overline{f}(u_l)=\overline{f}(u_r)$ otherwise (i.e. $u$ is an intersection node).  

In the second phase, for an arbitrary state $s\in \fitch(\rt)$ the algorithm constructs   a Fitch-extension $\overline{f}$ in the following top-down manner. We start with $\overline{f}(\rho)=s$. Suppose that $v$ is a child of $u$ for which $\overline{f}(u)$ is defined, then
\begin{equation}
\label{eq:fitch:extension}
\overline{f}(v)=
\begin{cases}
\overline{f}(u)  & \mbox{if $\overline{f}(u)\in \fitch(v)$,} \\
\mbox{any state in $\fitch(v)$} & \mbox{otherwise.}
\end{cases}
\end{equation}

Since each union node will contribute precisely one mutation for the extension $\overline{f}$ specified in Equation~(\ref{eq:fitch:extension}), each Fitch-extension is always minimum. (However, note that a minimum extension is not necessarily a Fitch-extension \cite{felsenstein2004inferring}.) The following observation, which we use later, is immediate from the second phase of Fitch's algorithm.


\begin{obs}
\label{obs:fitch2general}
Let $T$ be a rooted binary tree on $X$ and let $f$ be a character on $X$. Let $\rt$ be the root of $T$ and consider the Fitch map $\fitch$ induced by $f$. For each state $s\in \fitch(\rt)$, there exists a Fitch-extension $\overline{f}$ of $f$ such that $\overline{f}(\rt)=s$.
\end{obs}

\section{Chain reduction}
\label{sec:chain}

Let $T$ be an unrooted binary tree on $X$. For a leaf $x_i \in X$, let $p_i$ denote the internal node of $T$ adjacent to this leaf. Then, an ordered sequence $(x_1, \ldots, x_k)$ of $k$ taxa is called a  \emph{chain of length $k$} if $(p_1,p_2,\ldots,p_k)$ is a path in $T$. Note that here we allow that $p_1=p_2$ (i.e., $x_1$ and $x_2$ have a common parent) and/or  $p_{k-1}=p_k$ (i.e. $x_{k-1}$ and $x_{k}$ have a common parent): if at least one of these
situations occurs we say the chain is \emph{pendant}. (This is equivalent to definitions used
in earlier articles). A chain is \emph{common}  to $T_1$ and $T_2$ if it is a chain of both trees. 
Suppose $T_1$ and $T_2$ have a common chain $K =(x_1, \ldots, x_k)$ where $X(K)$ denotes the taxa in the chain and $k = |X(K)|\geq 5$.  Let $T'_1$, $T'_2$ be two new trees on $X' = (X \setminus X(K)) \cup \{x_1, x_2, x_{k-1}, x_{k}\}$ where $T'_1 = T_1|_{X'}$ and $T'_2 = T_2|_{X'}$. Then
we say that  $T'_1$, $T'_2$ have been obtained by \emph{reducing $K$ to length 4}.

\begin{thm}
\label{thm:reduc}
Let $T_1$ and $T_2$ be two unrooted binary trees on the same set of taxa $X$. Let $K$ be a common chain of length $k\geq 5$. Let $T'_1$ and $T'_2$ be the two trees obtained by reducing $K$ to length 4. Then $\dMP(T_1, T_2) = \dMP(T'_1, T'_2)$.
\end{thm}
\begin{proof}
Note that $\dMP(T'_1, T'_2) \leq \dMP(T_1, T_2)$ follows from Corollary 3.5 of \cite{dMP-moulton}, which
proves that for all $Y \subseteq X$, $\dMP( T_1 |_Y, T_2 |_Y ) \leq \dMP( T_1, T_2 )$. The inequality then follows from the definition of chain reduction.

It is considerably more involved to prove the claim that $\dMP(T'_1, T'_2) \geq \dMP(T_1, T_2)$ holds.

Without loss of generality, we may assume that $\dMP(T_1, T_2)>0$ (i.e., $T_1\not =T_2$) as otherwise the claim clearly holds. Note that this implies 
$X \not = X(K)$ and hence whenever $K$ is pendant
in a tree, at least one end of the chain is attached to the main part of the tree. 

We will prove the claim by considering the following three major cases: (I) the common chain is pendant in \emph{neither} tree, (II) the chain is pendant in preciesly one tree, and (III) the chain is pendant in both trees.\\ 
 \\
\textbf{I: Common chain is pendant in neither tree}\\
\\
 Let $f$ be an optimal character for $T_1$ and $T_2$
i.e. $|l_f(T_1) - l_f(T_2)| = \dMP(T_1, T_2)$. Assume without loss of generality that $l_f(T_1) < l_f(T_2)$, so $\dMP(T_1, T_2) = l_f(T_2) - l_f(T_1)$. 

\begin{figure}[h]
 \begin{center}
  \includegraphics[width=10cm]{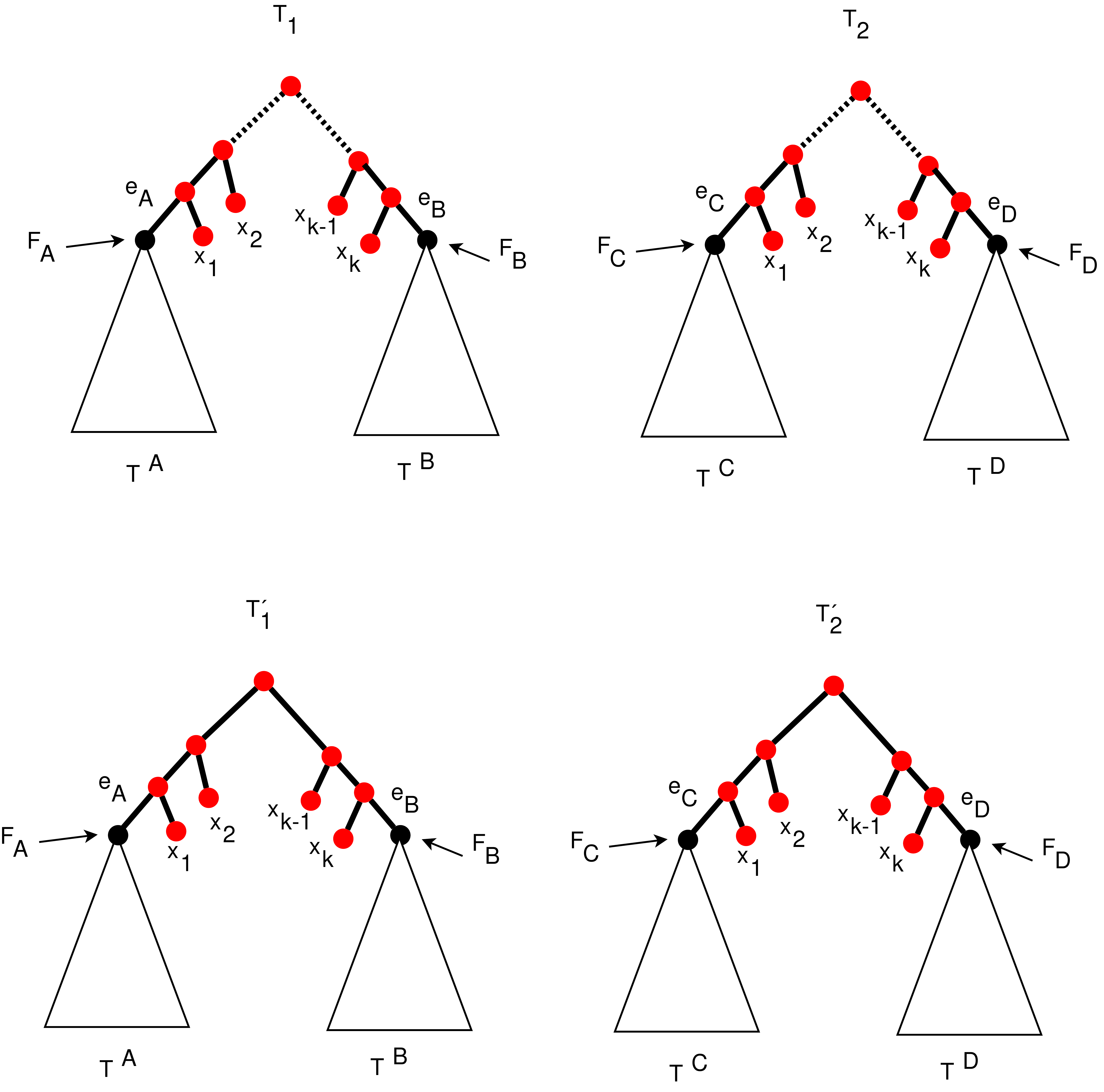}
 \end{center}
 \caption{The chain reduction as applied in the case when the common chain $K$ is pendant in neither tree. Note that in
$T_1$ and $T_2$ a dotted line is used to denote the taxa $\{x_3, \ldots, x_{k-2}\}$ which are removed by the chain
reduction. All the trees in the figure are unrooted, but for the purpose of proving correctness of the chain reduction we have shown them as rooted. $T'_1$ and $T'_2$ must be rooted exactly halfway along the chain, as shown. For $T_1$ and $T_2$ it is not so important where the tree is rooted as long as the root is in the same part of the chain in both trees.}
\label{fig:neitherpendant}
\end{figure}

Let $T^{A}, T^{B}, T^{C}, T^{D}$ refer to the 4 subtrees of $T_1, T_2$ shown in Figure \ref{fig:neitherpendant}. For $P \in \{A,B,C,D\}$, let $e_{P}$ refer to the edge incoming to the root of $T^{P}$;  let $X_{P}$ refer to the taxa in subtree $T^{P}$; let $f_{P}$ denote the character obtained
by restricting $f$ to $X_{P}$, and let $\fitch_{P}$ refer to the set of states assigned to the root of $T^{P}$ by the Fitch map induced by $f_{P}$. (Note that $X_A \cup X_B = X_C \cup X_D$.) For each tree $T \in \{T_1, T_2, T'_1, T'_2\}$ we define the \emph{chain region} of $T$ to be the set of edges incident to at least one red node (as shown in Figure \ref{fig:neitherpendant}).  Let $m_i$ $(i=1,2)$ be the number of union nodes among red nodes, which is the same as the number of mutations occuring in the chain region of $T_i$ for a Fitch-extension of $f$. Then, 
\begin{align*}
m_1&= l_f(T_1)- l_{f_A}(T^{A}) - l_{f_B}(T^{B})  ~\mbox{and} \\
m_2&=l_f(T_2) -l_{f_C}(T^{C}) - l_{f_D}(T^{D}). 
\end{align*}
 In addition, let $p=m_2-m_1$ and then we have
\begin{equation}
\label{eq:p}
\dMP(T_1, T_2) =  l_{f_C}(T^{C}) + l_{f_D}(T^{D}) -  l_{f_A}(T^{A}) - l_{f_B}(T^{B}) + p.
\end{equation}

First we shall show that $p\leq 2$. To this end, fix a Fitch-extension $\overline{f}_1$ of $f$ to $T_1$, and consider an extension
$\overline{f_2}$ of $f$ to $T_2$ obtained by combining a minimum extension of 
$f_C$ to $T^C$, a minimum extension of $f_D$ to $T^D$, and 
exactly mimicking $\overline{f_1}$ on the red nodes of $T_2$ (as indicated in Figure~\ref{fig:neitherpendant}).
 Then compared with $\overline{f}_1$, the extension $\overline{f}_2$ creates at most two new mutations  on the chain region (i.e.  edges $e_C$ and $e_D$). In other words,  we have 
$\ch(\overline{f}_2)\leq l_{f_C}(T^{C}) + l_{f_D}(T^{D})+(m_1+2)$. Together with $l_f(T_2) \leq \ch(\overline{f}_2)$ and $l_f(T_1)=l_{f_A}(T^{A}) + l_{f_B}(T^{B})+m_1$, this implies
\begin{eqnarray}
p&=&l_f(T_2) - l_f(T_1)- l_{f_C}(T^{C}) - l_{f_D}(T^{D}) +  l_{f_A}(T^{A}) + l_{f_B}(T^{B}) \nonumber \\
&=& l_f(T_2) -m_1- l_{f_C}(T^{C}) - l_{f_D}(T^{D})\nonumber \\
&\leq & \ch(\overline{f}_2)-m_1- l_{f_C}(T^{C}) - l_{f_D}(T^{D}) \nonumber \\
&\leq& 2.   \label{eq:p:up}
\end{eqnarray}

Next we show $p\geq 0$.  Consider a new (not necessarily optimal) character $f^*$ obtained from $f$ by reassigning all the taxa in $X(K)$ to a new state $s$ that does not appear anywhere on $X\setminus X(K)$. Considering Fitch-extensions of $f^{*}$ to $T_1$ and to $T_2$ we observe that $T_1$ and $T_2$ will both incur exactly 2 mutations in their chain regions, namely on edges $e_A$, $e_B$ and $e_C$, $e_D$, respectively. That is, we have
 \begin{equation}
 \label{eq:2:mutation}
 l_{f^*}(T_1)=l_{f_A}(T^{A}) +l_{f_B}(T^{B}) +2~~\mbox{and}~~ l_{f^*}(T_2)=l_{f_C}(T^{C}) +l_{f_D}(T^{D}) +2.
 \end{equation} 
Since the optimality of $f$ implies $l_f(T_2) - l_f(T_1)\geq l_{f^*}(T_2) - l_{f^*}(T_1)$, by Equation~(\ref{eq:2:mutation}) we have
\begin{eqnarray}
p&=&l_f(T_2) - l_f(T_1)- l_{f_C}(T^{C}) - l_{f_D}(T^{D}) +  l_{f_A}(T^{A}) + l_{f_B}(T^{B}) \nonumber \\
&\geq& l_{f^*}(T_2) - l_{f^*}(T_1)- l_{f_C}(T^{C}) - l_{f_D}(T^{D}) +  l_{f_A}(T^{A}) + l_{f_B}(T^{B}) \nonumber \\
&=& 0. \label{eq:p:low}
\end{eqnarray}

\bigskip
By Equation~(\ref{eq:p}),  the claim  $\dMP(T'_1, T'_2) \geq \dMP(T_1, T_2)$ will follow from 
\begin{equation}
\label{eq:p:subtree}
\dMP(T'_1, T'_2) \geq l_{f_C}(T^{C}) + l_{f_D}(T^{D}) -  l_{f_A}(T^{A}) - l_{f_B}(T^{B}) + p.
\end{equation}
Therefore, to establish main case ({\bf I}) it is sufficient to establish  Equation~(\ref{eq:p:subtree}), which will be done through case analysis on $p$. To shorten notation we will write $f[a,b,c,d]$ to denote the character on $X'$ obtained
from $f$ (which is a character on $X$) by leaving the states assigned to taxa in $X_A \cup X_B = X_C \cup X_D$ intact and assigning states
$a,b,c,d$ to $x_1, x_2, x_{k-1}, x_{k}$ respectively. (Occasionally we will manipulate
$f$ to obtain a new character $f^{*}$ also on $X$, and then the expression $f^{*} = f[a,b,\ldots,c,d]$ is overloaded to denote the reassignment of states to the taxa in the original
chain $K$, not the reduced chain.) Since $p$ is an integer with $0\leq p\leq 2$, we have the following three cases to consider.\\
\\
\noindent
\emph{Case 1:} $p=0$. Let $f' = f[s,s,s,s]$ where $s$ is a state that does not appear elsewhere. Then by the 
``both trees incurring exactly 2 mutations in their chain regions for Fitch-extensions'' reason used in the proof of Equation~(\ref{eq:2:mutation}), we have
$l_{f'}(T'_2) - l_{f'}(T'_1) = l_{f_C}(T^{C}) + l_{f_D}(T^{D}) -  l_{f_A}(T^{A}) - l_{f_B}(T^{B})$, from which Equation~(\ref{eq:p:subtree}) holds.\\
\\
\emph{Case 2:} $p=1$. We require a subcase analysis on $\fitch_A, \fitch_B, \fitch_C, \fitch_D$. 

\begin{enumerate}
\item[(i)] $\fitch_A \setminus \fitch_C \neq \emptyset$: Let $a \in \fitch_A \setminus \fitch_C$. Consider a state $s$, which is a state that does not appear elsewhere, and 
the character $f' = f[a,s,s,s]$. If
we consider Fitch-extensions of $f'$ on $T'_1$ and on $T'_2$, we
see that in $T'_1$ there are exactly 2 mutations incurred in the chain region, and in $T'_2$
exactly 3, and we are done, because we now have $l_{f'}(T'_1)=l_{f_A}(T^{A})+ l_{f_B}(T^{B})+2$ and $l_{f'}(T'_2)=l_{f_C}(T^{C})+ l_{f_C}(T^{C})+3$, so $\dMP(T'_1,T'_2) \geq l_{f'}(T'_2)-l_{f'}(T'_1)=\dMP(T_1,T_2)$. The latter equality is true because we are in the case where $p=1$. For brevity we henceforth speak of ``an $(i,j)$ situation'' when there are $i$ mutations in the chain region in tree $T'_1$ and $j$ in $T'_2$, so in this case we have a (2,3) situation.

\item[(ii)]  $\fitch_B \setminus \fitch_D \neq \emptyset$: This is symmetrical to the previous case.

\item[(iii)]  $(\fitch_A \subseteq \fitch_C) \wedge (\fitch_B \subseteq \fitch_D)$: This case cannot occur. Intuitively,
$T_2$ is ``less constrained'' than $T_1$ at the roots of the subtrees, so there is no way
that $T_1$ can use the chain region to save mutations relative to $T_2$. More formally,
consider a Fitch-extension $\overline{f_1}$ of $f$ to $T_1$. Then by definition $\overline{f_1}$  assigns a state $a$ from $\fitch_{A}$ to the root of $T_A$, and a state $b$ from $\fitch_{B}$ to the root of $T_B$ (where $a$ and $b$ are not necessarily different). Since $a\in \fitch_C$, by Observation~\ref{obs:fitch2general}, we fix a Fitch-extension $\overline{f}_C$ of 
$f_C$ to $T^C$ that maps the root of $T^C$ to $a$. Similarly, we fix  a Fitch-extension $\overline{f}_D$ of $f_D$ to $T^D$ that maps the root of $T^D$ to $b$. Now consider the extension $\overline{f_2}$ of $f$ to $T_2$ obtained by combining $\overline{f}_C$, $\overline{f}_D$, and 
exactly mimicking $\overline{f_1}$ for the red nodes of $T_2$. 
Then the number of mutations induced by $\overline{f_2}$ in the chain region of $T_2$ is exactly the same as that by $\overline{f_1}$ in the chain region of $T_1$. In other words, we have $\ch(\overline{f}_2)= l_{f_C}(T^{C}) + l_{f_D}(T^{D})+m_1$, from which we conclude that, if $(\fitch_A \subseteq \fitch_C) \wedge (\fitch_B \subseteq \fitch_D)$, then
\[
\dMP(T_1, T_2) =l_{f}(T_2)-l_{f}(T_1) \leq \ch(\overline{f}_2)-l_{f}(T_1)
=  l_{f_C}(T^{C}) + l_{f_D}(T^{D}) -  l_{f_A}(T^{A}) - l_{f_B}(T^{B}).
\]
In particular, this shows $p \leq 0$, a contradiction. We will re-use (slight variations of) this argument repeatedly to show that certain subcases cannot occur. For brevity we will refer to it as the \emph{less constrained roots} argument.

\end{enumerate}

\noindent
\emph{Case 3:} $p=2$.  Then we have the following two subcases to consider.
\begin{enumerate}
\item[(i)]  $(\fitch_A \setminus \fitch_C \neq \emptyset) \wedge (\fitch_B \setminus \fitch_D \neq \emptyset)$: Let $a \in \fitch_A \setminus \fitch_C$ and $b \in \fitch_B \setminus \fitch_D$. We take character $f' = f[a,s,s,b]$
where $s$ does not occur elsewhere. This is a $(2,4)$ situation.

\item[(ii)]  ($\fitch_A \subseteq \fitch_C) \vee (\fitch_B \subseteq \fitch_D)$: By a variant of the \emph{less constrained roots} argument, we know this case cannot
occur as otherwise it leads to $p\leq 1$, a contradiction.
\end{enumerate}
\noindent
\textbf{II: Common chain is pendant in exactly one tree}\\

\begin{figure}[h]
 \begin{center}
  \includegraphics[width=10cm]{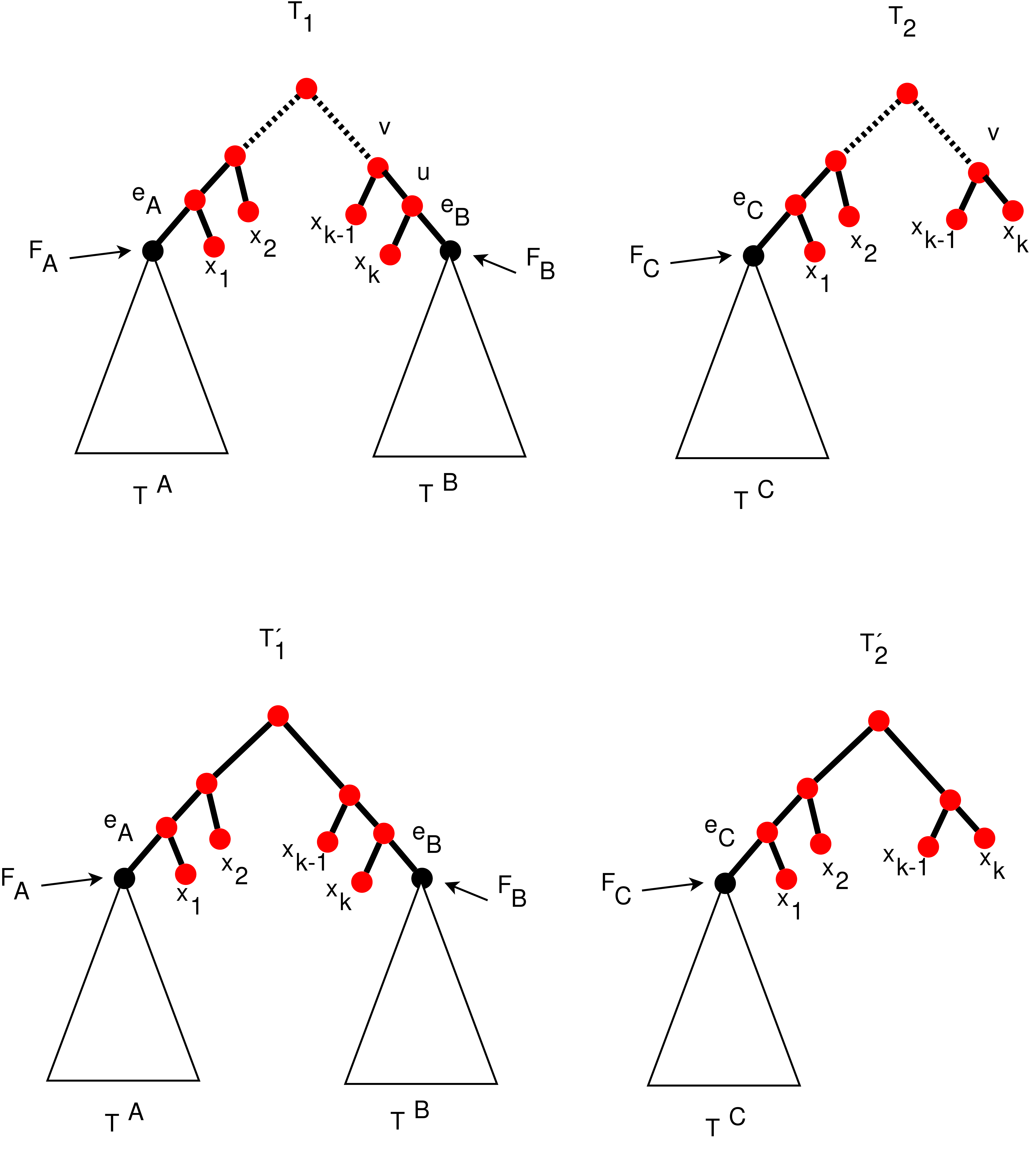}
 \end{center}
\caption{The situation when the common chain is pendant in exactly one tree.}
\label{fig:onependant}
\end{figure}

Without loss of generality we assume that $K$ is pendant in $T_2$ and that the situation is as described in Figure \ref{fig:onependant}. Let $f$ be an optimal character. Then we have the following two cases.

\medskip
\noindent
\emph{Case 1:} In this first case $l_f(T_1) < l_f(T_2)$, so
$\dMP(T_1, T_2) = l_f(T_2) - l_f(T_1)$. As in Equation~(\ref{eq:p}) we have,
\begin{equation}
\label{eq:p1}
\dMP(T_1, T_2) =  l_{f_C}(T^{C})  -  l_{f_A}(T^{A}) - l_{f_B}(T^{B}) + p.
\end{equation}
In this case, $p \leq 1$ because of the usual mimicking construction (i.e. copying the
states allocated to the red nodes in $T_1$, to $T_2$) used in the proof of Equation~(\ref{eq:p:up}). That is, at most 1 extra mutation incurs in $T_2$ (i.e. on the edge $e_{C})$\footnote{
Here the mimicking construction must deal with a slight technicality: node $u$ in $T_1$ (see Figure \ref{fig:onependant}) does not exist in $T_2$. However,  simply ignoring $u$ in this case (and elsewhere mapping $v$ to $v$) has the desired
effect: if there is a mutation on edge $(v,x_{k})$ in $T_2$ then there must have been at least one mutation on the edges $(v,u)$ and $(u,x_{k})$ in $T_1$.}.
 On the other hand $p \geq 0$ follows from an argument similar to that for proving Equation~(\ref{eq:p:low}). That is,  we can always relabel $f$
to a new character $f^{*} = f[a,s,\ldots,s,b]$ where $a \in \fitch_A$, $b \in \fitch_B$ and $s$ is a state that does not appear elsewhere. This is either a $(2, 2)$ or a $(2,3)$ situation, proving that $p \geq 0$. Hence, in Equation~(\ref{eq:p1}),
we have $p \in \{0,1\}$, and hence it remains to prove that
$$\dMP(T'_1, T'_2) \geq l_{f_C}(T^{C})  -  l_{f_A}(T^{A}) - l_{f_B}(T^{B}) + p$$
holds, which will be done by considering the following two subcases. 

\begin{enumerate}
\item[(i)]  $p=0$: Suppose first $\fitch_A \not \subseteq \fitch_C$. Let $a \in \fitch_A \setminus \fitch_C$. Note that $a \not \in \fitch_B$  because otherwise the character $f^{*} = f[a,a,...,a,a]$ would lead to a $(0,1)$ situation, contradicting $p=0$. This implies that the character $f' = f[a,a,a,a]$ is a $(1,1)$ situation
and we are done. So suppose next $\fitch_A \subseteq \fitch_C$.  If $\fitch_A \cap \fitch_B \neq \emptyset$ then let $a \in \fitch_A \cap \fitch_B$. Clearly $a \in \fitch_C$. Taking character $f'= f[a,a,a,a]$ yields
a $(0,0)$ situation and we are done. Otherwise, $\fitch_A \cap \fitch_B = \emptyset$. In this
situation, let $a \in \fitch_A \cap \fitch_C$ and let $b \in \fitch_B$. (Clearly, $a \neq b$). Consider character
$f' = f[a,a,b,b]$. This is a $(1,1)$ situation and we are done.

\item[(ii)]   $p=1$: Suppose $\fitch_A \not \subseteq \fitch_C$. Let $a \in \fitch_A \setminus \fitch_C$. If $a \in \fitch_B$ then we take $f'=f[a,a,a,a]$. This is a $(0,1)$ situation
and we are done. If $a \not \in \fitch_B$, then let $b \neq a$  be an arbitrary element of
$\fitch_B$. We take $f' = f[a,a,b,b]$, this is a $(1,2)$ situation and we are done. The
only subcase that remains is $\fitch_A \subseteq \fitch_C$, but this cannot happen by the
\emph{less constrained roots} argument.
\end{enumerate}

\medskip
\noindent
\emph{Case 2:} We have $l_f(T_2) < l_f(T_1)$, so $\dMP(T_1, T_2) = l_f(T_1) - l_f(T_2)$.
In such a case we have
\begin{equation}
\dMP(T_1, T_2) =  l_{f_A}(T^{A}) +  l_{f_B}(T^{B}) -  l_{f_C}(T^{C})  +  p.
\end{equation}
We have $p \leq 2$, by the usual mimicking argument, but this time the red nodes in $T_1$
copy their states from $T_2$ and not the other way round. (Nodes $u$ and $x_k$ in $T_1$ should both be assigned
the state that is assigned to $x_k$ in $T_2$). Also, $p \geq 1$ because we can relabel $f$ to a new character $f^{*} = f[s,s,\ldots,s,s]$ where $s$ is a state that
does not appear elsewhere. This is a $(2,1)$ situation. Hence, $p \in \{1,2\}$.
and hence it remains to prove that
$$\dMP(T'_1, T'_2) \geq   l_{f_A}(T^{A}) + l_{f_B}(T^{B})-l_{f_C}(T^{C}) + p$$
holds, which will be done by considering the following two subcases. 

\begin{enumerate}

\item[(i)]  $p=1$. Take $f' = f[s,s,s,s]$, where $s$ is a state that does not appear elsewhere. This is a $(2,1)$ situation, and we are done.

\item[(ii)]   $p=2$. Suppose $\fitch_C \not \subseteq \fitch_A$. Consider $f' = f[c,c,s,s]$ where $s$ is a state that does not occur elsewhere and $c \in \fitch_C \setminus \fitch_A$. This is a $(3,1)$ situation and we are done. The only remaining case is $\fitch_C \subseteq \fitch_A$: but this is not possible
by the \emph{less constrained roots} argument.

\end{enumerate}
\noindent
\textbf{III: Common chain is pendant in both trees}\\
\begin{figure}[h]
 \begin{center}
  \includegraphics[width=10cm]{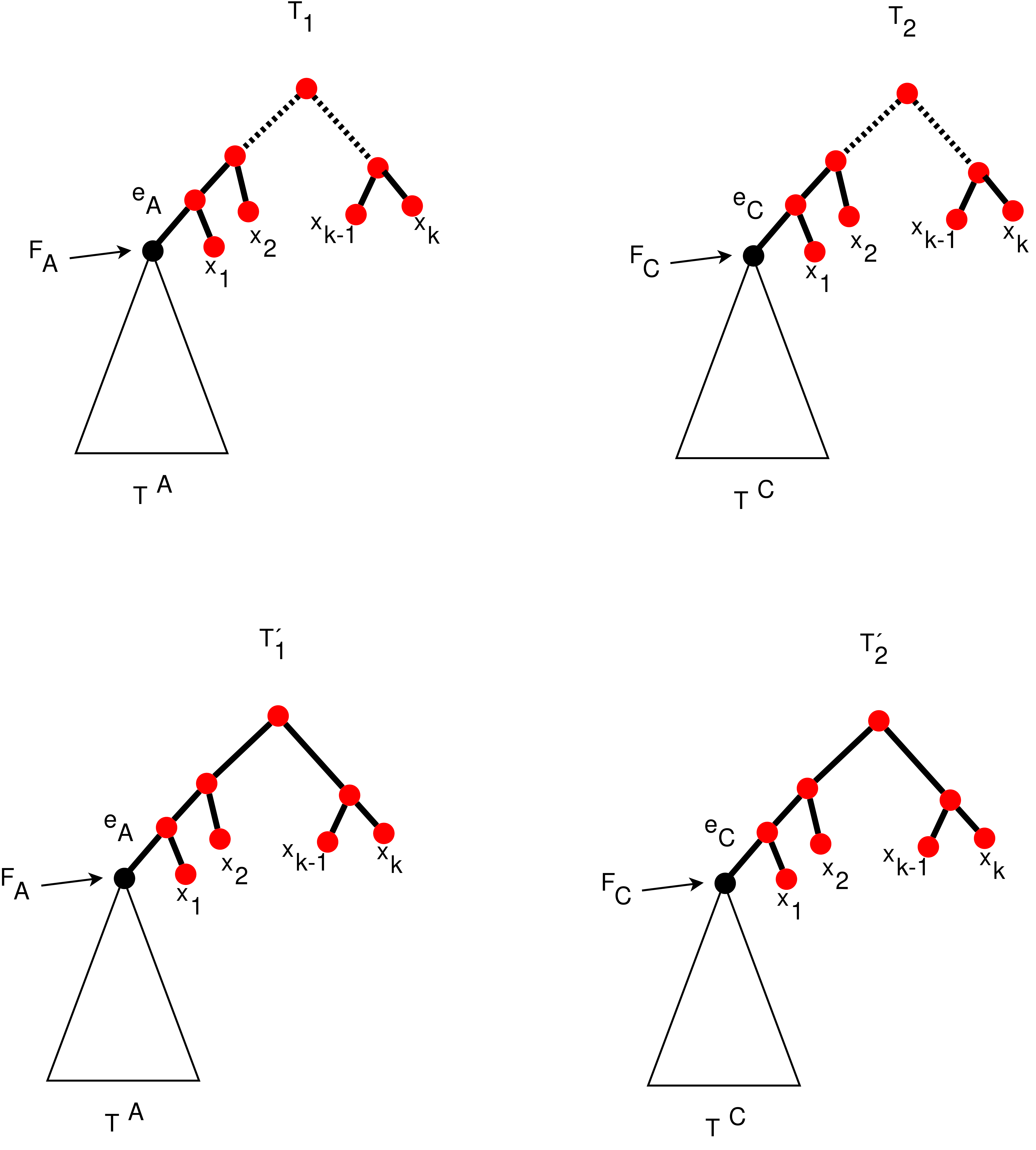}
 \end{center}
\caption{The situation when the common chain $K$ is pendant in both trees and the chain
is oriented in the same direction in both trees (relative to the point of contact with the rest of the tree).}
\label{fig:onependantsync}
\end{figure}
\begin{figure}[h]
 \begin{center}
  \includegraphics[width=10cm]{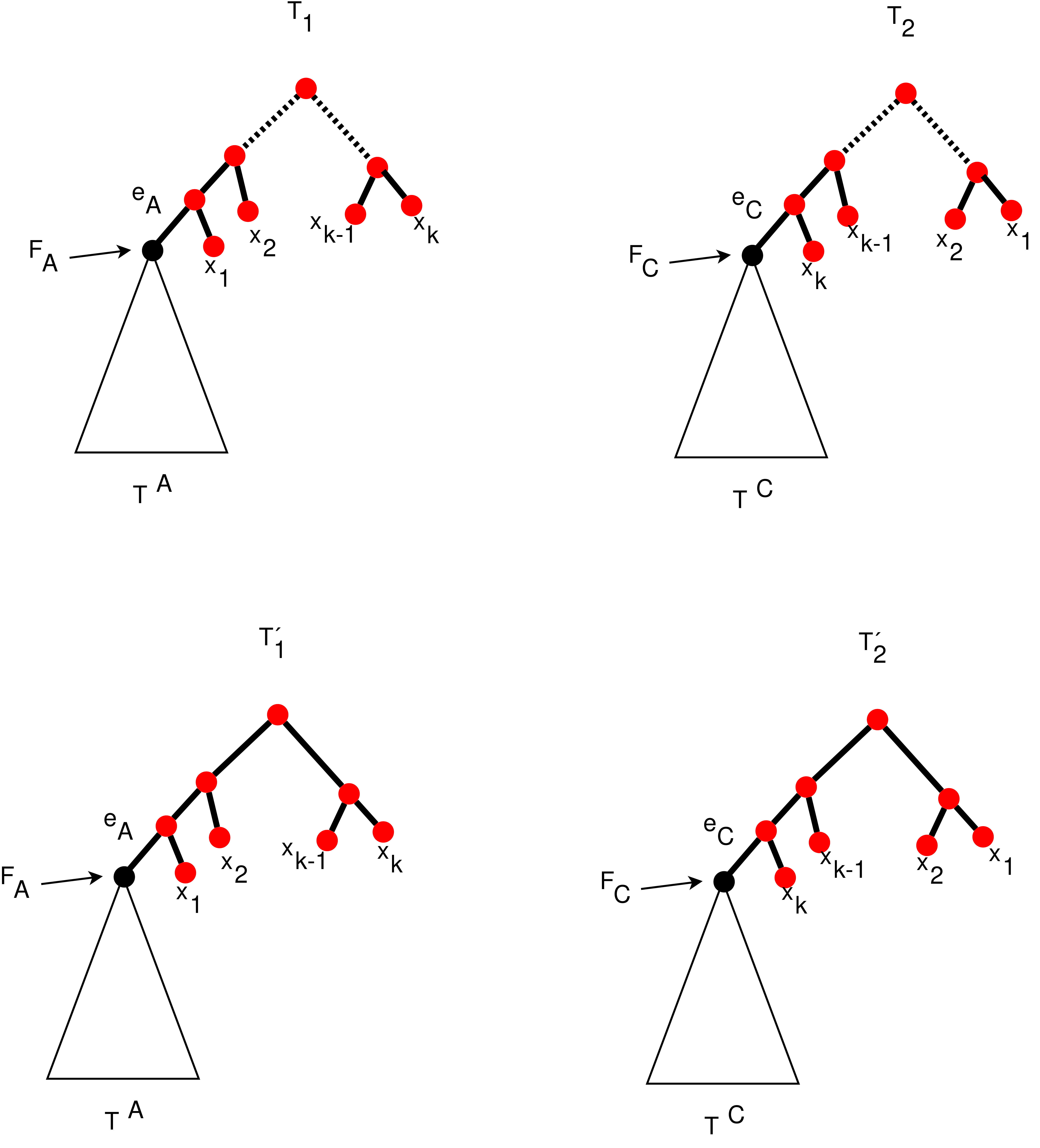}
 \end{center}
\caption{The situation when the common chain is pendant in both trees and the chain
is oriented in different directions in the two trees (relative to the point of contact with the rest of the tree).}
\label{fig:onependantunsync}
\end{figure}

There are two main situations here: the chains are oriented in the same direction (Figure \ref{fig:onependantsync}),
and the chains are oriented in the opposite direction (Figure \ref{fig:onependantunsync}). Whichever situation occurs, we can
assume without loss of generality that $l_f(T_1) < l_f(T_2)$, so
$\dMP(T_1, T_2) = l_f(T_2) - l_f(T_1)$. As in Equation~(\ref{eq:p}) we have,
\begin{equation}
\label{eq:p:III}
\dMP(T_1, T_2) =  l_{f_C}(T^{C})  -  l_{f_A}(T^{A})  + p.
\end{equation}
Note that we have $p \geq 0$ by the familiar trick of assigning all the taxa in $X(K)$ a state that does not occur elsewhere and $p \leq 1$ by the mimicking construction. It remains to show that $$\dMP(T'_1, T'_2) \geq l_{f_C}(T^{C})  -  l_{f_A}(T^{A})  + p$$
holds, which can be done by considering the following three cases:

\medskip
\noindent
\emph{Case 1:} $p=0$. In this case we can just take $f' = f[s,s,s,s]$ where $s$ is a state that does not appear elsewhere: this is a $(1,1)$ situation, and we are done.

\medskip
\noindent
\emph{Case 2:}  $p=1$ and we are in the same-direction situation. Observe that
$\fitch_A \subseteq \fitch_C$ cannot hold by the \emph{less constrained roots}
argument. So $\fitch_A \not \subseteq \fitch_C$. Let $a \in \fitch_A \setminus \fitch_C$. Consider
the character $f' = f[a,s,s,s]$ where $s$ is a state that does not appear elsewhere. This
is a $(1,2)$ situation and we are done. 

\medskip
\noindent
\emph{Case 3:} $p=1$ and we are in the opposite-direction situation. Then take
$f' = f[a,a,s,s]$ where $a \in \fitch_A$ and $s$ is a state that does not occur elsewhere. This is a $(1,2)$ situation (note that here
we are exploiting the fact that $K$ is reversed in $T_2$ relative to $T_1$, the
status of $\fitch_C$ is not relevant here), so we are done.
\end{proof}

Note that Theorem \ref{thm:reduc} is in some sense best possible, since reducing common chains to length 3 can
potentially alter $\dMP$; see Figure \ref{fig:bestpossible} for a concrete example. Here
$\dMP(T_1,T_2) \geq 2$ (due to character $abcdefgh=00001111$) and $\dMP(T_1,T_2) \leq \dTBR(T_1,T_2) \leq 2$ - due to the agreement forest $\{\{a,b\}, \{c,d,e,f\}, \{g,h\}\}$ - so
$\dMP(T_1,T_2) = \dTBR(T_1,T_2)=2$. However, $\dMP(T'_1,T'_2) = 1$  (achieved by
character $abdefgh=0000111$); the fact that $\dMP(T'_1, T'_2) \leq 1$ can be verified
computationally.

The chain reduction can easily be performed in polynomial time, and it can be applied at most a polynomial number of times because each application of the reduction reduces the number of taxa by at least 1. Hence, we obtain the following corollary.
\begin{cor}
Let $T_1$ and $T_2$ be two unrooted binary trees on the same set of taxa $X$. Then
it is possible to transform $T_1, T_2$ to $T'_1, T'_2$ in polynomial time such that all
common chains in $T'_1, T'_2$ have length at most 4 
and $\dMP(T_1, T_2) = \dMP(T'_1, T'_2)$.
\end{cor}

\begin{figure}[h]
 \begin{center}
  \includegraphics[width=10cm]{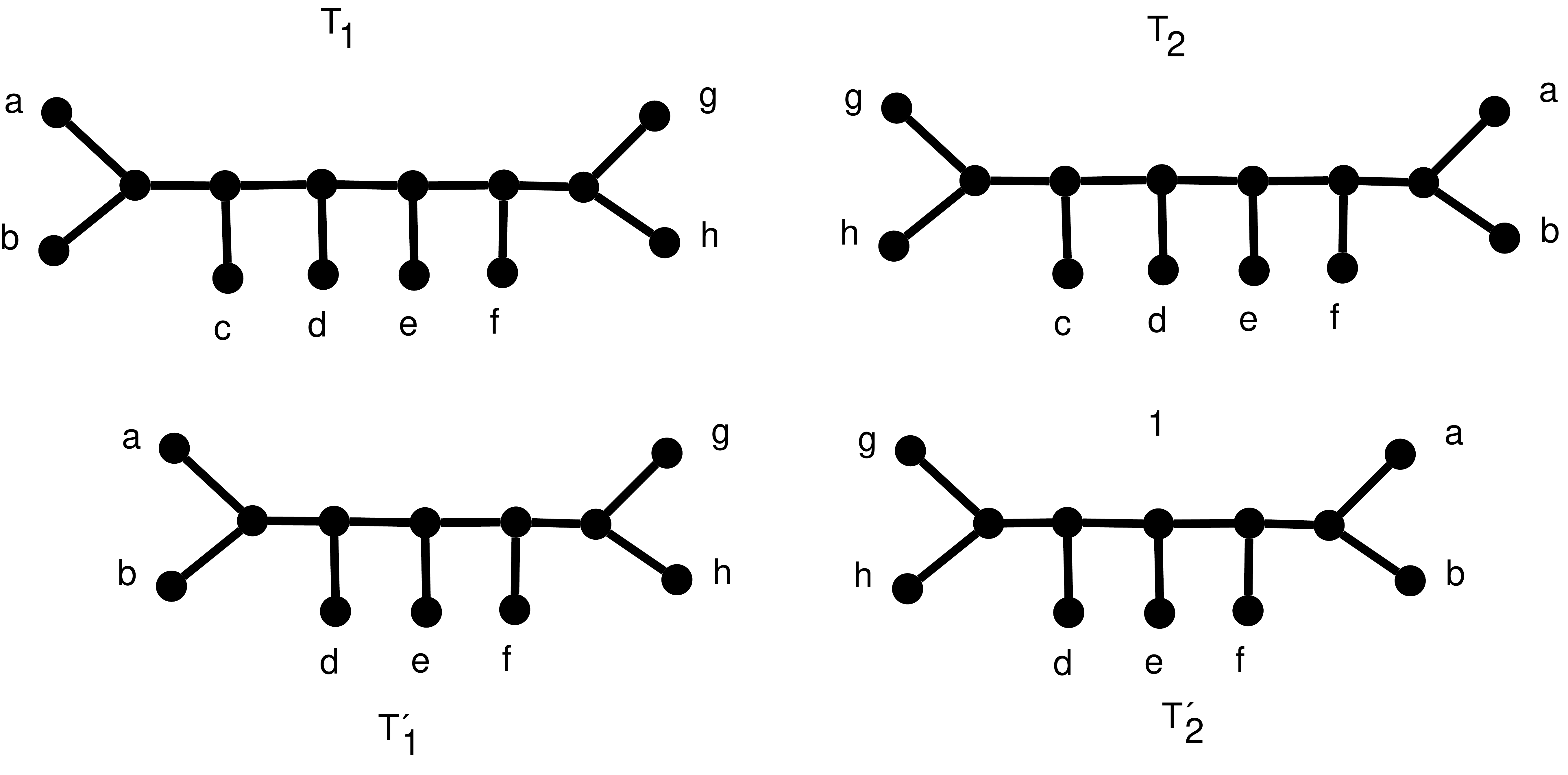}
 \end{center}
 \caption{  Here $\dMP(T_1, T_2)=2$, while $\dMP(T'_1, T'_2) = 1$. This
shows that reducing common chains to length 3 does not preserve $\dMP$. Note that
$\dTBR(T_1,T_2)=\dTBR(T'_1,T'_2)=2$, because $\dTBR$ \emph{is} preserved under
reduction of chains to length 3 \cite{allen01}.}
\label{fig:bestpossible}
\end{figure}

\section{A generalized subtree reduction}
\label{sec:subtree}

Let $T_1$ and $T_2$ be two unrooted binary trees on a set of taxa $X$. A \emph{split} $A|B$ (on $X$) is simply a bipartition of $X$ i.e. $A \cap B = \emptyset$, $A \cup B = X$, $A, B \neq \emptyset$. For a phylogenetic tree $T$ on $X$, we say that edge $e$ \emph{induces} a split
$A|B$ if, after deleting $e$, $A$ is the subset of taxa appearing in one connected component
and $B$ is the subset of taxa appearing in the other. 

\begin{figure}[h]
 \begin{center}
  \includegraphics[width=10cm]{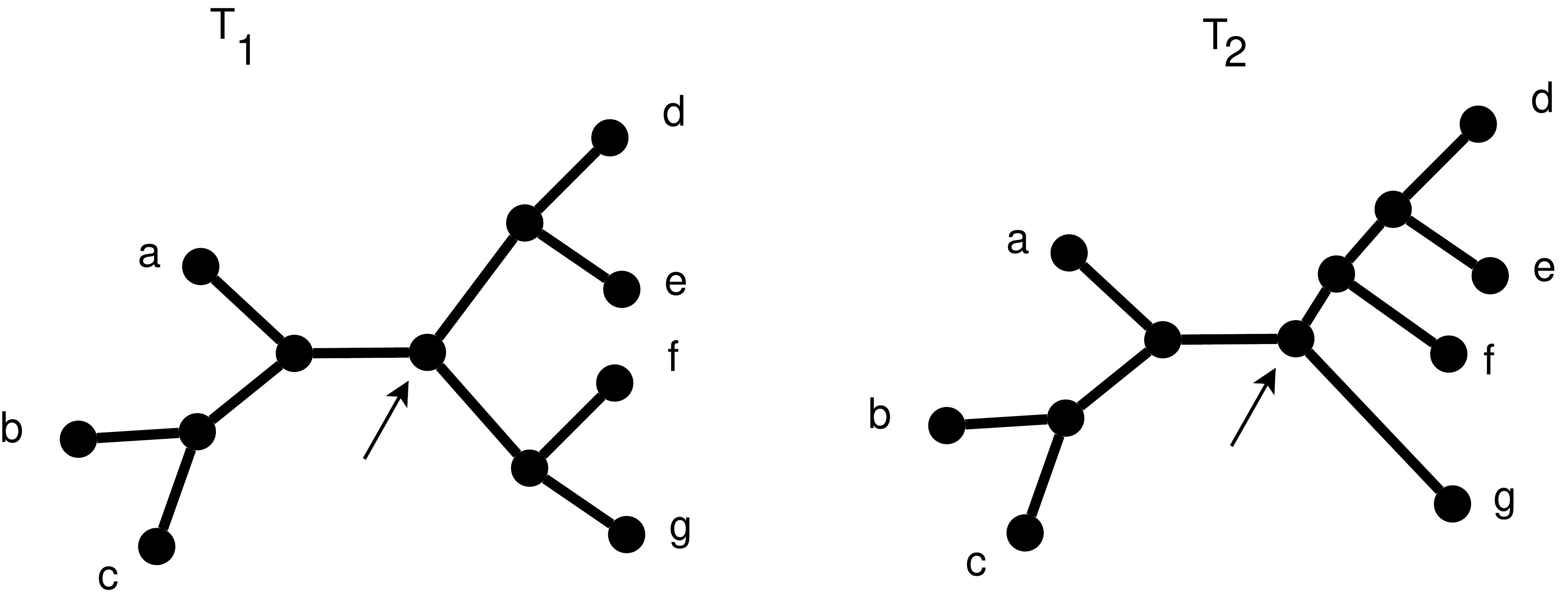}
 \end{center}
 \caption{Here $T_1$ and $T_2$ have a common pendant subtree on $\{a,b,c\}$, and
a common pendant subtree ignoring root location on $\{d,e,f,g\}$. Note that $T_1|_{\{d,e,f,g\} }= T_2|_{\{d,e,f,g\}}$ but
the rooted variants $(T_1|_{\{d,e,f,g\}})^\rt$ and $(T_2|_{\{d,e,f,g\}})^\rt$ are not equal because they have different
root locations (indicated here with an arrow).} 
\label{fig:subtreeExample}
\end{figure}

Consider $X' \subset X$. We say that $T_1$ and $T_2$ have a \emph{common pendant subtree ignoring root location (i.r.l.) on} $X'$ if (1) for $i \in \{1,2\}$, $T_i$ contains an
edge $e_i = \{u_i, v_i\}$ such that $e_i$ induces a split 
$(X \setminus X')| X'$ in $T_i$ and (2) $T_1 | X' = T_2 | X'$. Now, assume without loss of generality
that for $i \in \{1,2\}$, $v_i$ is the endpoint of edge $e_i$ that is closest to taxon set $X'$. The
node $v_i$ can be used to ``root'' $T_i|_{X'}$, yielding a rooted binary phylogenetic
tree on $X'$ which we denote $(T_i|_{X'})^\rt$.  If $T_1$ and $T_2$ have the additional property
that $(T_1|_{X'})^\rt = (T_2|_{X'})^\rt$ (where here the equality operator is acting over
rooted trees), then we say that $T_1$ and $T_2$ have a \emph{common pendant subtree} on $X'$. Clearly, a common pendant subtree on $X'$ is also a common
pendant subtree i.r.l. on $X'$, but the other direction does not necessarily hold. The
following reduction takes both types of subtrees into account.\\
\\
%
\textbf{Generalized subtree reduction:} Let $T_1$ and $T_2$ be two unrooted binary trees on $X$. Let $X'$ be a subset of $X$ such that $|X'| \geq 2$. (If $T_1 = T_2$ and $X' = X$,
then clearly $\dMP(T_1, T_2)=0$ so $T_1$ and $T_2$ can simply be replaced with a single
taxon. We henceforth assume $X' \subset X$). Suppose $T_1$ and $T_2$ have a common pendant subtree i.r.l. on $X'$. We construct a reduced pair of trees $T'_1$ and $T'_2$
as follows. If $T_1$ and $T_2$ have a common pendant subtree on $X'$,
we are in the \emph{traditional} case. If they do not, and $|X'| \geq 4$, we are in the
\emph{extended} case. If we are in neither case, the generalized subtree reduction does not apply.
\begin{itemize}
\item \textbf{Traditional case.} Let $T'_1 = T_1 |_{(X \setminus X') \cup \{x\}}$ and
$T'_2 = T_2 |_{(X \setminus X') \cup \{x\}}$ where $x \in X'$. (This
is the ``traditional'' subtree reduction, as described in e.g. \cite{allen01} and \cite{dMP-moulton}.)

\item \textbf{Extended case}.  Without loss of generality let $x,y,z$ be distinct taxa in $X'$ such that in $(T_1|_{X'})^\rt$, $x$ and $y$ are on one side of the root, and $z$ on the other, while in $(T_2|_{X'})^\rt$
$x$ and $z$ are on one side of the root, and $y$ on the other. These taxa always
exist because $(T_1|_{X'})^\rt \neq (T_2|_{X'})^\rt$. We let
$T'_1 = T_1 | (X \setminus X') \cup \{x,y,z\}$ and $T'_2 = T_2 | (X \setminus X') \cup \{x,y,z\}$.
\end{itemize}

Note that the reduction can easily be applied in polynomial time. Also, each application
reduces the number of taxa by at least one, so if the reduction is applied repeatedly it
will stop after at most polynomially many iterations.

\begin{thm}
\label{thm:gensubtree}
Let $T_1$ and $T_2$ be two unrooted binary trees on the same set of taxa $X$. Suppose that  $T'_1$ and $T'_2$ are two reduced trees obtained by applying the generalized
subtree reduction to $T_1$ and $T_2$. Then $\dMP(T'_1, T'_2) = \dMP(T_1, T_2)$.
\end{thm}
\begin{proof}
If the traditional case applies, then the result is immediate from \cite{dMP-moulton}. Hence,
let us assume that we are in the extended case. As in the proof of Theorem
\ref{thm:reduc}, $\dMP(T'_1, T'_2) \leq \dMP(T_1, T_2)$ follows from Corollary 3.5 of \cite{dMP-moulton}. It remains to show $\dMP(T'_1, T'_2) \geq \dMP(T_1, T_2)$.  To this end, we may further assume $\dMP(T_1, T_2)>0$ as otherwise the theorem clearly holds.

\begin{figure}[h]
 \begin{center}
  \includegraphics[width=10cm]{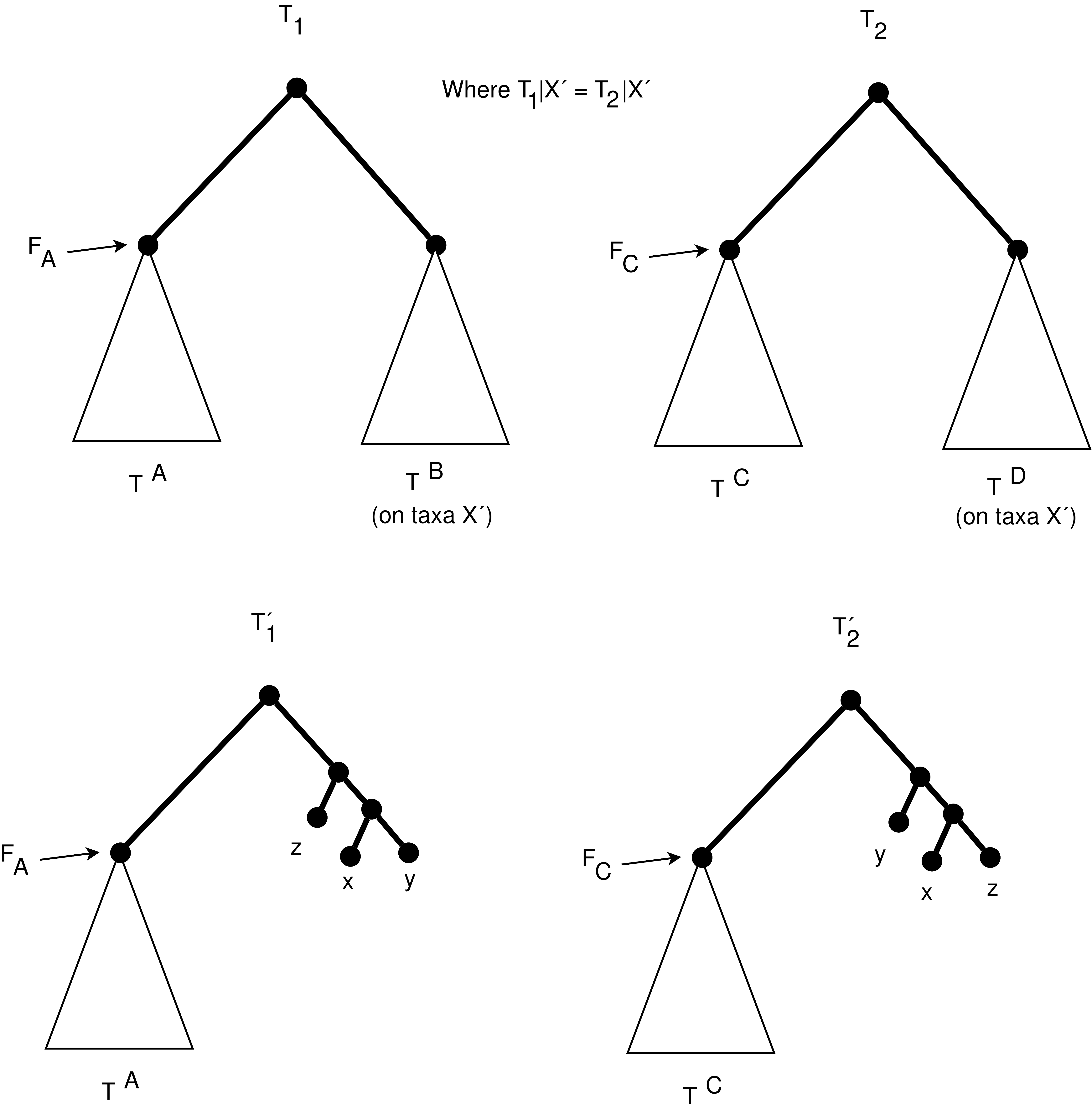}
 \end{center}
 \caption{The generalized subtree reduction as it behaves in its extended case. That is,
when $|X'| \geq 4$, $T_1|_{X'} = T_2|_{X'}$ but $(T_1|_{X'})^\rt \neq (T_2|_{X'})^\rt$.}
\label{fig:subtree}
\end{figure}

Let $f$ be an optimal character (in the usual sense) for $T_1$ and $T_2$
i.e. $|l_f(T_1) - l_f(T_2)| = \dMP(T_1, T_2)$. 
Let $T^{A}, T^{B}, T^{C}, T^{D}$ refer to the 4 subtrees of $T_1, T_2$ shown in Figure \ref{fig:subtree}. For $P \in \{A,B,C,D\}$, let $X_{P}$ refer to the taxa in subtree $T^{P}$. Here $X_B = X_D = X'$, $X_A = X_C = X \setminus X'$ and
$T_1 | X' = T_2 | X'$. That is, $T^{B}$ and $T^{D}$ are identical subtrees assuming we ignore the
point at which each subtree is connected to the rest of its tree. As indicated in the figure, we
root $T_1$ and $T_2$ (to put them in an appropriate form for  Fitch-extensions) by subdividing
the edge that connects each pendant subtree to the rest of the tree. Let $f_{P}$ denote the character obtained
by restricting $f$ to $X_{P}$, and let $\fitch_{A}$ refer to the set of states assigned to the root of $T^{A}$ by the Fitch map induced by $f_{A}$.

For $i\in\{1,2\}$, let $m_i=0$ if the root of $T_i$ is an intersection node, and $m_i=1$ otherwise (i.e. the root is a union node). Then we have 
\begin{align*}
l_f(T_1) & = l_{f_A}(T^{A}) + l_{f_B}(T^{B}) + m_1,~~~ \mbox{and} \\
l_f(T_2) & = l_{f_C}(T^{C}) + l_{f_D}(T^{D}) + m_2.
\end{align*}
Note that we also have $l_{f_B}(T^{B}) = l_{f_D}(T^{D})$ because $T^{B}$ and $T^{D}$ are (from an unrooted perspective) identical.

In the remainder of the proof we shall assume that $l_f(T_1) < l_f(T_2)$, as the other case $l_f(T_1) > l_f(T_2)$ is symmetrical. Let $p=m_2-m_1$. Then we have
\begin{align}
\dMP(T_1, T_2) & = l_f(T_2) - l_f(T_1)\nonumber\\
&= (l_{f_C}(T^{C}) + l_{f_D}(T^{D}) + m_2) - (l_{f_A}(T^{A}) + l_{f_B}(T^{B}) + m_1) \nonumber \\
&= l_{f_C}(T^{C}) -  l_{f_A}(T^{A}) + (m_2 - m_1) \nonumber \\
&=  l_{f_C}(T^{C}) -  l_{f_A}(T^{A}) + p. \label{eq:tree:p}
\end{align}
Now we claim $p \in \{0, 1\}$. To see this, by definition of $p$ it suffices to show that $p \geq 0$: Indeed, fix a state  $s$ \steven{that} is not used elsewhere and consider the character $f^{*}$ obtained from modifying $f$ by assigning all the taxa in $X'$ to the state $s$; then we have $l_{f^{*}}(T_2) = l_{f_C}(T^{C}) + 1$ and $l_{f^{*}}(T_1) = l_{f_A}(T^{A}) + 1$, from which we can conclude that $\dMP(T_1, T_2) \geq l_{f^{*}}(T_2) - l_{f^{*}}(T_1) =  l_{f_C}(T^{C}) -  l_{f_A}(T^{A})$, and hence $p\geq 0$.

In order to show $\dMP(T'_1, T'_2) \geq \dMP(T_1, T_2)$, by Equation~(\ref{eq:tree:p}) it suffices to show that 
\begin{equation}
\label{eq:subtree:p}
\dMP(T'_1, T'_2) \geq   l_{f_C}(T^{C}) -  l_{f_A}(T^{A}) + p.
\end{equation}
To shorten notation we will write $f[a,b,c]$ to denote the character on $(X \setminus X') \cup \{x,y,z\}$ obtained
from $f$ by leaving the states assigned to taxa in $X_A = X_C = (X \setminus X')$ intact and assigning states
$a,b,c$ to $x, y, z$ respectively. Since $p \in \{0, 1\}$,  we have the following two cases:\\
\\
\emph{Case 1: $p=0$.} Let $f' = f[s,s,s]$ where $s$ is a state that does not appear elsewhere.
Then $l_{f'}(T'_1) = l_{f_A}(T^{A}) + 1$ and $l_{f'}(T'_2) = l_{f_C}(T^{C}) + 1$. This implies 
$$
\dMP(T'_1, T'_2)\geq l_{f'}(T'_2)-l_{f'}(T'_1) = l_{f_C}(T^{C})-l_{f_A}(T^{A}), 
$$  from which Equation~(\ref{eq:subtree:p}) follows and we are done.\\
\\
\emph{Case 2: $p=1$.} Let $a \in \fitch_A$ and let $s$ be a state that does not appear
elsewhere. Consider $f' = f[s,s,a]$. Observe that $l_{f'}(T'_1) = l_{f_A}(T^{A}) + 1$
and $l_{f'}(T'_2) = l_{f_C}(T^{C}) + 2$, so we are done by an argument similar to that in \emph{Case 1}. 
\end{proof}

Note that the generalized subtree reduction could be used to replace the ``pendant in both trees'' case of the chain reduction. If the chains are oriented the same way they will be reduced to a single taxon (using the traditional case of the subtree reduction) and if they are oriented in opposite direction they will be reduced to a subtree of size 3 (using the extended case of the subtree reduction). We have described the chain reduction and the generalized subtree reduction
separately to emphasize that in terms of correctness the two reductions are independent of each other.

\section{Parameterized algorithms}

As stated in the introduction, combining Theorems \ref{thm:gensubtree} and 
\ref{thm:reduc} with the kernelization in \cite{allen01} and the exponential-time algorithm
for $\dMP$ described in \cite{kelkFibonacci}, yields the following theorem:\\
\\
\textbf{Theorem \ref{thm:fpttbr}}.
\emph{Let $T_1$ and $T_2$ be two unrooted binary 
trees on the same set of species $X$. Computation
of $\dMP(T_1,T_2)$ is fixed parameter tractable 
in parameter $\dTBR = \dTBR(T_1, T_2)$.  More 
specifically, $\dMP(T_1,T_2)$ can be computed 
in time $O( \phi^{c \cdot \dTBR} \cdot \emph{poly}(|X|) )$ where $\phi \approx 1.618...$
is the golden ratio and $c \leq 112/3$.}\\
\\
We close the main part of the paper by observing that a purely theoretical version of
Theorem \ref{thm:fpttbr} can be obtained via Courcelle's Theorem~\cite{courcelle1990,Arnborg91}. \steven{A few further definitions are first necessary}. \steven{Given an undirected graph $G = (V,E)$, a \emph{bag} is simply a subset of $V$. A \emph{tree decomposition} of $G$ consists of a tree $T_{G} = (V(T_G), E(T_G))$ where $V(T_G)$
is a collection of bags such that the following holds: (1) every node of
$V$ is in at least one bag; (2) for each edge $\{u,v\} \in E$, there exists some bag that
contains both $u$ and $v$; (3) for each node $u \in V$, the bags that contain $u$ induce
a connected subtree of $T_G$. The \emph{width} of a tree decomposition is equal to the
cardinality of its largest bag, minus 1. The \emph{treewidth} of a graph $G$, denoted $tw(G)$, is equal
to the minimum width, ranging over all possible tree decompositions of $G$ \cite{bodlaender1994tourist, BodlaenderK10}. A tree with at least one edge has treewidth 1.} \steven{The \emph{display graph} of two unrooted binary phylogenetic trees $T_1$ and $T_2$, both on the same set of taxa $X$, is the graph $D(T_1, T_2)$ obtained by identifying leaves that are labelled with the same taxon \cite{bryant2006compatibility}. See Figure \ref{fig:displaygraph} for an example.} \steven{A formal description of \emph{Monadic Second Order Logic (MSOL)} is beyond the scope of this article; we refer to
\cite{kelk2015} for an introduction relevant to phylogenetics. Informally, it is a type of logic used to describe properties of graphs, in which both universal (``for all'') and existential (``there exists'') quantification are permitted over (subsets of) nodes and (subsets of) edges.}

\begin{figure}[h]
 \begin{center}
  \includegraphics[width=6cm]{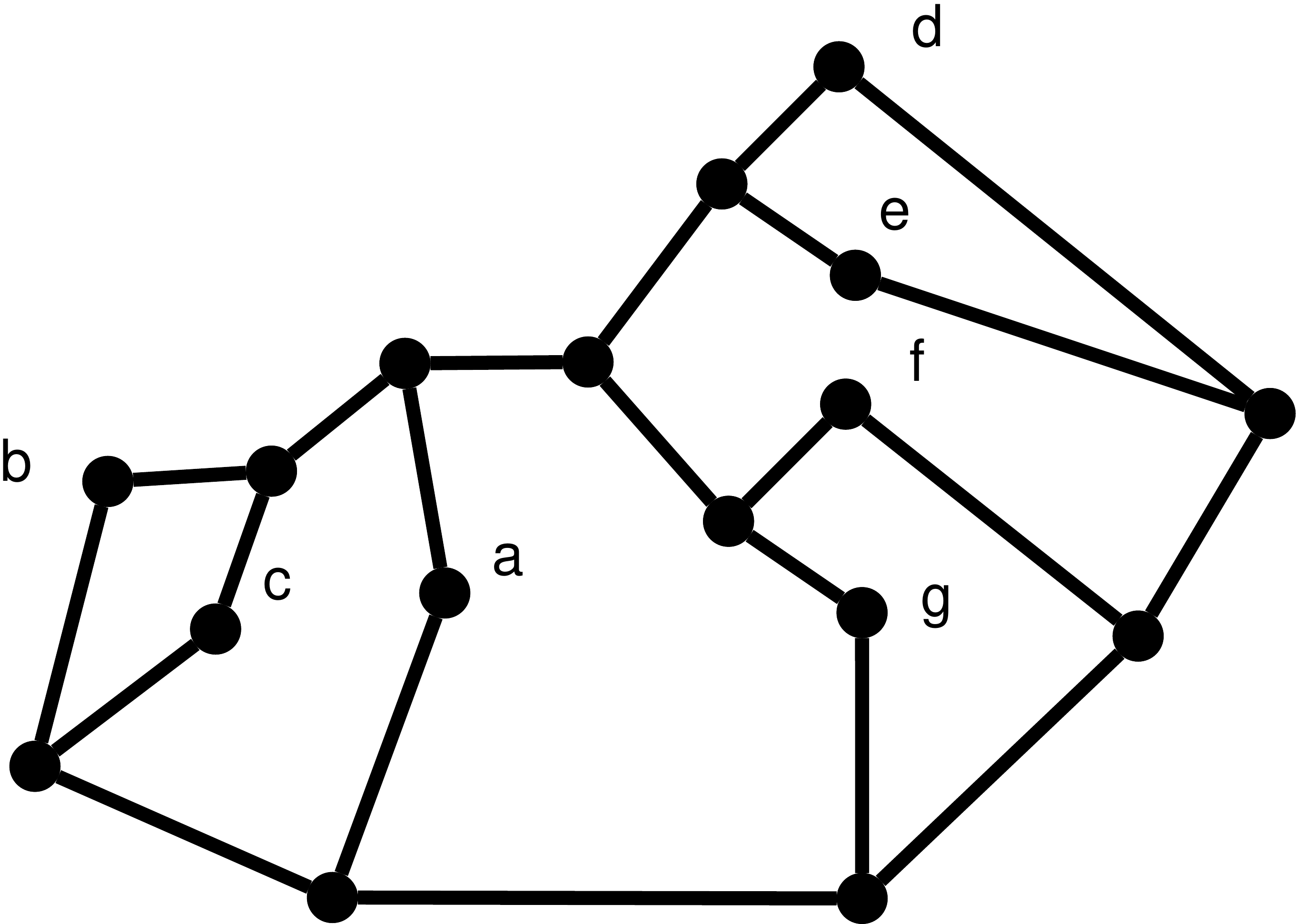}
 \end{center}
 \caption{\kelk{The display graph $D(T_1, T_2)$ obtained from the trees $T_1$ and $T_2$ shown
in Figure \ref{fig:subtreeExample}. The treewidth of this graph is 3, and
$\dMP(T_1, T_2) = \dTBR(T_1,T_2)=1$.}} 
\label{fig:displaygraph}
\end{figure}

\begin{remark}
Let $T_1$ and $T_2$ be two unrooted binary 
trees on the same set of species $X$. Via \steven{Monadic Second Order Logic}  \steven{(MSOL)} it can be shown
that computation
of $\dMP(T_1,T_2)$ is possible in time $g( \dTBR ) \cdot \emph{poly}(|X|)$ where
 $\dTBR = \dTBR(T_1, T_2)$ and $g$ is some computable function that depends only
on $\dTBR$.
\end{remark}
We do not give \steven{explicit} details \kelk{of this alternative FPT proof} since the argument is extremely indirect and does not in any
sense lead to a practical algorithm: the function $g$ is astronomical. However, for completeness we sketch the overall idea. In \cite{kelk2015} it is shown that computation of
$d^{2}_\mathrm{MP}$ (the variant of {\dMP} in which characters are restricted to at most 2 states)
is FPT in parameter \dTBR. The core insight there is (i) \steven{the display graph $D(T_1,T_2)$} has treewidth bounded by a function of $\dTBR$ and (ii) Fitch's algorithm can be modelled in a static fashion by guessing an optimal character and subsequently guessing the Fitch maps
induced by that character in the two trees (including whether each node is a union or intersection node). This naturally requires that the internal nodes of the trees are partitioned into $O( 2^{|\stateset|} )$ subsets, where as usual $\stateset$ is the set of states used by the optimal character. From \cite{boes2015} it is known that there always exists an optimal character in which $|\stateset| \leq 7 \dMP - 5$. Now, there is a polynomial-time 3-approximation for computation of $\dTBR$ (see \cite{chen2015parameterized} for a recent overview), so running such an algorithm yields a value $t$ such that $\dTBR \leq t \leq 3 \dTBR$. Combining with the fact that $\dMP \leq \dTBR$  \cite{dMP-fischer,dMP-moulton}, it follows that $7t-5$ is an 
upper bound on the number of states required to encode an optimal character for \dMP. Also, 
$7t-5$ is clearly bounded by a function of $\dTBR$, which means that the resulting
sentence of \steven{MSOL} has a length that is bounded by a (admittedly highly exponential) function of $\dTBR$. The result then follows from the \steven{optimization} variant of Courcelle's Theorem \steven{known as EMS} \steven{which is} described by Arnborg et al. in \cite{Arnborg91}.

\section{Discussion and open problems}
\label{sec:disc}

A major open question is whether the two reductions 
discussed in this article (the chain
reduction and the generalized subtree reduction) are together sufficient
to obtain a kernel for $\dMP$.  That is, after applying
the rules repeatedly until they can no longer be 
applied, is it true that the
number of taxa in the resulting instance is 
bounded by some function of $\dMP$? If answered affirmatively,  
this would prove that computation of $\dMP$ 
is FPT in its most natural 
parameterization, 
namely $\dMP$ itself, which would mean that $\dMP$ can be computed in time
$f(\dMP) \cdot \text{poly}(|X|)$ for some computable function $f$ that depends only on
$\dMP$.

Note that,  if it can be 
shown that $\dTBR \leq g( \dMP )$ for some function $g$ 
that depends only on $\dMP$, then the desired FPT result 
will follow automatically from Theorem~\ref{thm:fpttbr}. 
In \cite{dMP-moulton} it is claimed that $\dTBR \leq 2 \dMP$,
and while the claim itself is not known to be false, the proof is incorrect. In fact, at the present time we do not know how to prove $\dTBR \leq g( \dMP )$ for \emph{any} $g(.)$, even when $g$ is extremely fast-growing. Relatedly, we do not even know how to compute $\dMP$ in time $O( |X|^{f( \dMP )})$ for any computable function $f$ that only depends on $\dMP$. Running times of this latter form (which are algorithmically weaker results than FPT) are trivial for $\dTBR$ and most other
tree distances.

This is intriguing because, although tree-pairs are known where $\dTBR = 2 \dMP$ (see e.g. Figure \ref{fig:bestpossible}),
empirical tests suggest that $\dMP$ and $\dTBR$
are in practice often \emph{extremely} close. The following simple experiment highlights this.
For each $n \in \{10,15,20,25\}$ and $t \in \{0.2n, 0.4n, 0.6n, 0.8n\}$ we generated 500 tree pairs, where the first tree is generated uniformly at random from the space of unrooted binary trees on $n$ taxa, and the second tree is obtained from the first by randomly applying at most $t$ TBR moves. We computed $\dMP$ using the algorithm described in \cite{kelkFibonacci} and $\dTBR$ using an ad-hoc Integer Linear Programming (ILP) formulation. The ILP formulation is the running time bottleneck, limiting us to 25 taxa. For every $(n, t)$
parameter combination, at most 1 tree-pair was observed that had $\dMP = \dTBR - 2$ (and this was the largest difference we observed). In Table \ref{tab:exp},
the first number is the $\%$ of the 500 tree pairs that had $\dMP = \dTBR$,
and the second number is the $\%$ of the tree pairs where $\dMP \geq \dTBR - 1$.

\begin{table}
\begin{center}
\begin{tabular}{|c||c|c|c|c|}
\hline
$n,t$ & $0.2n$ & $0.4n$ & $0.6n$ & $0.8n$ \\
\hline
\hline
\textbf{10} & 99.8, 100 & 96.2, 100 & 91.6, 100 & 89, 100\\
\hline
\textbf{15} & 99.2, 100 & 96.4, 99.8 & 94, 100 & 87, 100\\
\hline
\textbf{20} & 99.8, 100 & 97.6, 100 & 90.2, 99.8 & 87.4, 100\\
\hline
\textbf{25} & 99.8, 100 & 96.2, 100 & 91, 99.8 & 77.9, 100\\
\hline
\end{tabular} 
\caption{Percentage of the 500 tree-pairs on $n$ taxa (and at most $t$ TBR moves apart) in which
$\dMP = \dTBR$, and $\dMP \geq \dTBR -1$, respectively.}
\label{tab:exp}
\end{center}
\end{table}

Despite these empirical observations there are some clues that $\dMP$ and $\dTBR$ might ultimately have a rather different combinatorial structure. Consider the following construction.
In \cite{dMP-kelk} it is shown, for every integer $k \geq 2$, how to construct a (rooted) tree-pair $T_1, T_2$
such that $\dMP(T_1, T_2) = \dTBR(T_1, T_2) = 4k$ and,
\[
(\max_{f}  l_f(T_2) - l_f(T_1)) \geq (\max_{f}  l_f(T_1) - l_f(T_2)) + (k+1).
\]
(As usual $f$ in this context ranges over all characters). Such tree-pairs
are considered ``asymmetric''. Fix an arbitrary constant $k \geq 2$ and let $T_1, T_2$ be such a tree-pair, where $X$ denotes their set of taxa. Produce exact copies of $T_1, T_2$ on a new set
of taxa $X'$, and call these trees $T'_1, T'_2$. Connect $T_1$ and $T'_2$ together at their roots
by an edge - call this new tree $T_1:T'_2$ - and do the same for $T'_1$ and $T_2$ to obtain
the new tree $T'_1:T_2$. Both $T_1:T'_2$ and $T'_1:T_2$ are on taxa set $X \cup X'$ and
both have a common split $X|X'$.

It is straightforward to show that, due to the fact that $T_1:T'_2$ and $T'_1:T_2$ have been
constructed by joining asymmetric trees together in ``antiphase'', the following holds:
\begin{align*}
 \dMP(T_1:T'_2, T'_1:T_2) &\leq 4k + (4k - (k+1)) + 1\\
&= 7k.
\end{align*}
On the other hand, it is not too difficult to show (using agreement forests) that
\[
\dTBR(T_1:T'_2, T'_1:T_2) = 8k.
\]
Given that $k$ can be chosen arbitrarily, the difference between $7k$ and $8k$ can be
made arbitrarily large. This emphasizes that $\dMP$ and $\dTBR$ behave rather differently with regard to common splits. It also shows that if a tree-pair $T_1, T_2$ has a common split $P|Q$,  $\dMP(T_1, T_2)$ can (at least in an additive sense) be arbitrarily smaller than
$\dMP(T_1|P, T_2|P) + \dMP(T_1|Q, T_2|Q)$.

\begin{figure}[h]
 \begin{center}
  \includegraphics[width=10cm]{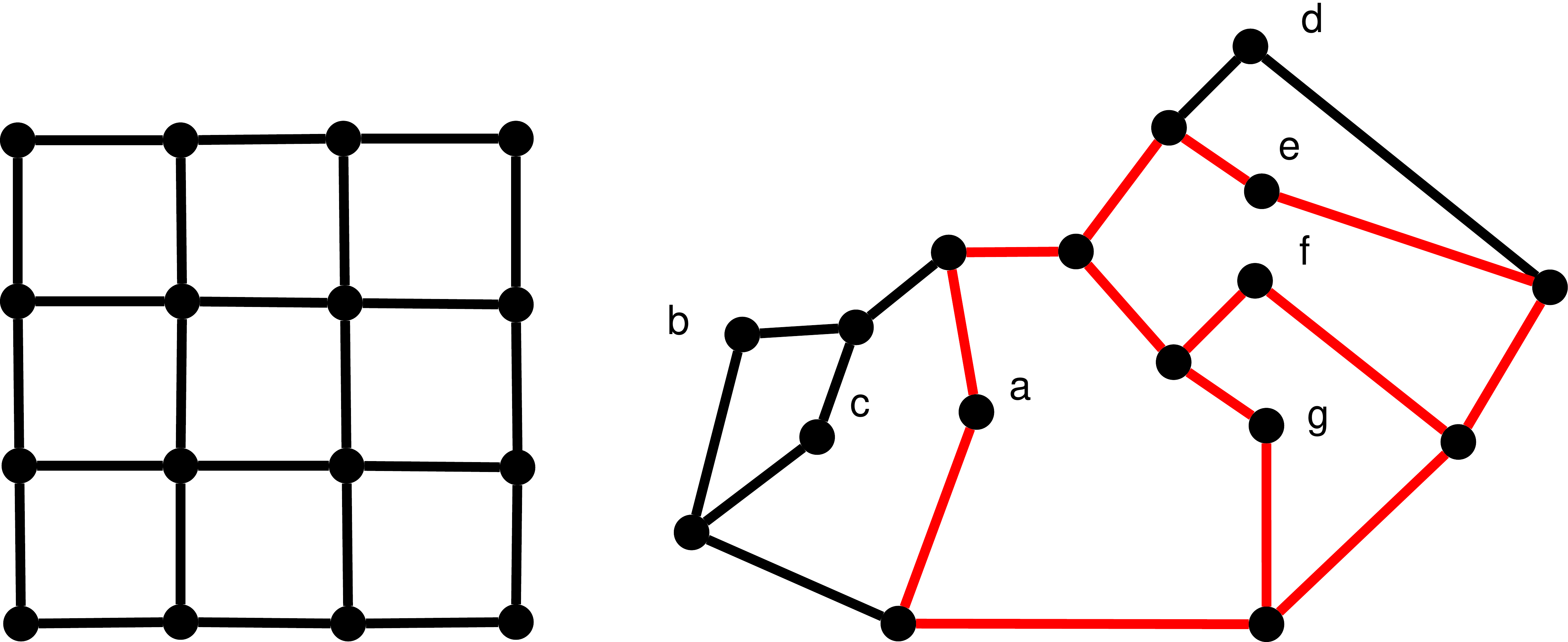}
 \end{center}
 \caption{\kelk{Left: a $4 \times 4$ grid. Right: The red edges show the embedding of a $K_4$ minor (which is itself a minor of the $3 \times 3$ grid) within the display graph
$D(T_1,T_2)$ from Figure \ref{fig:displaygraph}. The presence of this $K_4$ minor
establishes that the treewidth of the display graph is at least 3, and thus $\dMP(T_1, T_2)>0$.}} 
\label{fig:k4minor}
\end{figure}

\steven{Computation of $\dMP$ also touches on a number of structural issues relevant to algorithmic graph theory. In the MSOL approach described in the previous section both the length of the logical sentence, and the treewidth of the display graph, are bounded by a function of $\dTBR$. It is natural to ask whether bounds in terms of $\dMP$, rather than $\dTBR$, could be obtained because this would prove that $\dMP$ 
is FPT in its most natural 
parameterization (independently of the exact relationship between $\dMP$ and $\dTBR$). To bound the length of the sentence by a function of $\dMP$ it will be necessary
to identify a polynomial-time computable upper bound on $|\stateset|$ (the number of states used by some optimal character) that is bounded by a function of $\dMP$. This is a challenging question, albeit one that is tied closely to the very specific combinatorial structure of $\dMP$.} 

\steven{Establishing an $f( \dMP )$ bound on the treewidth of the display graph (for some function $f$) is, however, fundamental, in the following sense. An undirected graph $H$ is a \emph{minor} of an undirected graph $G$ if $H$ can be obtained from $G$ by deleting nodes, deleting edges and
contracting edges \cite{diestel2010}. The $n \times n$ \emph{grid} graph is (as its name suggests) simply the graph on $n^2$ nodes corresponding to the $n \times n$ square grid (see Figure \ref{fig:k4minor} for an example). From the \emph{grid minor theorem} it is well-known that if a graph has
treewidth $\geq k$, it has a grid minor of size at least $h(k) \times h(k)$ for a function $h$
that grows at least polylogarithmically quickly as a function of $k$ \cite{robertson1986graph, robertson1994quickly} (for more recent, stronger bounds on $h$ see \cite{chuzhoy2015excluded}). Hence, to prove that
the treewidth of the display graph is bounded by some function of $\dMP$ it is sufficient to
prove that, as grid minors in the display graph become larger and larger, $\dMP$ must also grow. The example of the $3 \times 3$ grid minor is illustrative (see Figure \ref{fig:k4minor}). If the display graph contains a $3 \times 3$ grid minor, it must also contain a $K_4$ minor (the complete undirected graph on 4 nodes), since $K_4$ is a minor of the $3 \times 3$ grid. Two compatible (i.e. $\dMP = 0$) phylogenetic trees induce display graphs of treewidth (at most) 2 \cite{bryant2006compatibility,grigoriev2015}, and graphs of treewidth at most 2 are characterized by the absence of $K_4$ minors. Hence, the presence of a $3 \times 3$ grid minor in the display graph implies $\dMP > 0$.}

\steven{Intuitively it seems plausible that larger grid minors will induce ever larger incongruencies between the two trees, thus driving $\dMP$ further up.
However, as demonstrated in \cite{grigoriev2015} formalizing such an intuition is a formidable task, since the embeddings of the minors can ``weave'' between the two trees in a difficult to analyse fashion. Indeed, this raises the question whether, and under which circumstances, the presence of (grid) minors in the display graph $D(T_1, T_2)$ can be translated into phylogenetic-topological statements about $T_1$ and $T_2$. This intersects with an emerging literature
at the interface of algorithmic graph theory and phylogenetics (see e.g. \cite{ grigoriev2015,kelk2015,Vakati2015337} and references therein).}

\section{Acknowledgements}

We thank Olivier Boes for helpful discussions. \kelk{We also thank the editor and the two anonymous referees for their constructive suggestions. SK and TW acknowledge the support of London Mathematical Society grant SC7-1516-05.}
\bibliographystyle{plain}
\bibliography{chain_2016_09_06}

\end{document}